\pdfoutput=1
\documentclass[sigconf]{acmart}

\usepackage[subtle]{savetrees}

\usepackage{needspace}
\usepackage{notation}
\usepackage{subcaption}
\usepackage{algorithmic}
\usepackage{algorithm}
\usepackage{amsthm}
\usepackage{array}
\usepackage{textcomp}
\usepackage{stfloats}
\usepackage{url}
\usepackage{verbatim}
\usepackage{graphicx}
\usepackage{booktabs}  
\usepackage{hyperref}  
\usepackage{todonotes}
\usepackage{listings}

\usepackage{xspace}
\usepackage{tikz}
\usetikzlibrary{positioning}
\usetikzlibrary{shapes}

\theoremstyle{plain}
\newtheorem{theorem}{Theorem}[section]

\theoremstyle{definition}
\newtheorem{definition}[theorem]{Definition}
\newtheorem{example}[theorem]{Example}

\newtheorem{remark}[theorem]{Remark}

\newcounter{example}

\usepackage{amssymb}      
\usepackage{stmaryrd}
\usepackage{multicol}
\usepackage{pdflscape}
\usepackage{longtable}
\usepackage{booktabs}
\usepackage{enumitem}

\usepackage{tcolorbox}

\newcommand{\CASE}[1]{\STATE \textbf{case} #1\textbf{:} \begin{ALC@g}}
\newcommand{\ENDCASE}{\end{ALC@g}}

\newcommand{\DEFAULT}{\STATE \textbf{default:} \begin{ALC@g}}
\newcommand{\ENDDEFAULT}{\end{ALC@g}}
\newcommand{\DEFAULTLINE}[1]{\STATE \textbf{default:} }

\title{An~STREL-based~Formulation~of~Spatial~Resilience~in Cyber-Physical Systems}

\author{Zeyu Zhang}
\email{zeyu.zhang.2@stonybrook.edu}
\affiliation{
  \institution{
  Stony Brook University}
  \city{Stony Brook, NY}
  \country{USA}}
\author{Hongkai Chen}
\email{hkchen@ie.cuhk.edu.hk}
\affiliation{
  \institution{
  The Chinese University of Hong Kong}
  \city{Hong Kong SAR}
  \country{China}}
\author{Nicola Paoletti}
\email{nicola.paoletti@kcl.ac.uk}
\affiliation{
  \institution{
  King’s College London}
  \city{London}
  \country{UK}}
\author{Shan Lin}
\email{shan.x.lin@stonybrook.edu}
\affiliation{
  \institution{
  Stony Brook University}
  \city{Stony Brook, NY}
  \country{USA}}
\author{Scott A. Smolka}
\email{sas@cs.stonybrook.edu}
\affiliation{
  \institution{
  Stony Brook University}
  \city{Stony Brook, NY}
  \country{USA}}

\newcommand{\wfun}{W}
\newcommand{\nextto}[3]{#1\stackrel{#2}{\mapsto}#3}
\newcommand{\route}{\tau}

\newcommand{\ssign}{\mathbf{s}}

\newcommand{\lserv}{\Sigma}

\newcommand{\fmon}{\mathbf{m}}

\newcommand{\reach}{\mathcal{R}}
\newcommand{\escape}{\mathcal{E}}

\newcommand{\sr}{S}
\newcommand{\spresv}{\sigma}
\DeclareMathOperator*{\minre}{\mathop{\min}_{\mathit{re}}}
\DeclareMathOperator*{\maxre}{\mathop{\max}_{\mathit{re}}}
\newcommand{\precre}{\mathrel{\prec_{\mathit{re}}}}
\newcommand{\succre}{\mathrel{\succ_{\mathit{re}}}}
\newcommand{\eqre}{\mathrel{=_{\mathit{re}}}}

\definecolor{lavander}{cmyk}{0,0.48,0,0}
\definecolor{violet}{cmyk}{0.79,0.88,0,0}
\definecolor{burntorange}{cmyk}{0,0.52,1,0}

\copyrightyear{2026}
\acmYear{2026}
\setcopyright{acmlicensed}\acmConference[HSCC' 26]{Submission of the 29th ACM International Conference on Hybrid Systems: Computation and Control}{May 11--14, 2026}{Saint Malo, France}
\acmBooktitle{Submission of the 29th ACM International Conference on Hybrid Systems: Computation and Control (HSCC' 26), May 11--14, 2026, Saint Malo, France}
\acmPrice{}
\acmDOI{}
\acmISBN{}

\begin{document}

\begin{abstract}
Resiliency is the ability of a system to quickly recover from a violation (recoverability) and avoid future violations for as long as possible (durability). In the spatial setting, recoverability and durability (now known as \emph{persistency}) are measured in units of distance.  Like its temporal counterpart, spatial resiliency is of fundamental importance for Cyber-Physical Systems (CPS) and yet, to date, there is no widely agreed-upon formal treatment of spatial resiliency. 

We present a formal framework for reasoning about spatial resiliency in CPS. Our framework is based on the spatial fragment of Spatio-Temporal Reach and Escape Logic (STREL), which we refer to as SREL. In this framework, spatial resiliency is given a syntactic characterization in the form of a \emph{Spatial Resiliency Specification} (SpaRS).  An atomic predicate of SpaRS is called an S-atom. Given an arbitrary SREL formula $\varphi$, distance bounds $d_1, d_2$, the S-atom of $\varphi$, $S_{d_1, d_2} (\varphi)$, is the SREL formula $\neg\varphi\reach_{[0,d_1]} (\varphi \reach_{[d_2, +\infty)}\varphi)$, specifying that recovery from a violation of $\varphi$ occurs within distance $d_1$ (\emph{recoverability}), and subsequently that $\varphi$ be maintained along a route for a distance greater than $d_2$ (\emph{persistency}). S-atoms can be combined using spatial STREL operators, allowing one to express composite resiliency specifications; e.g., multiple S-atoms must hold, or multiple locations must satisfy an S-atom.
We define a quantitative semantics for SpaRS in the form of a \emph{Spatial Resilience Value} (SpaRV) function $\spresv$ and prove its soundness and completeness w.r.t.\ SREL's Boolean semantics.
The $\spresv$-value for $S_{d_1,d_2}(\varphi)$ is a set of non-dominated $(rec, per)$ pairs, quantifying recoverability and persistency, given that some routes may offer better recoverability while some may provide better persistency. In addition, we design algorithms to evaluate SpaRV for SpaRS formulas. Finally, two case studies demonstrate the practical utility of our approach.
\end{abstract}


\maketitle

\section{Introduction}


Resilience is commonly understood as an ability to recover from or adjust easily to adversity or change~\cite{dictionary_resiliency}. In ~\cite{chen2022resilience}, a framework for temporal resilience in Cyber-Physical Systems (CPS) is presented, which takes into account how quickly a system recovers from a property violation (\emph{recoverability}), and how long post-recovery the system maintains the property in question (\emph{durability}).

Nonetheless, many CPS are also inherently \emph{spatial}: they operate in physical environments where resources, hazards, and operational conditions are distributed across space.  Correspondingly, once a violation of a property $\varphi$ occurs at a location, spatial resilience requires the system to be able to relocate to another location at which
$\varphi$ holds (within a spatial constraint), and subsequently visit consecutive spatial locations at which $\varphi$ holds (also spatially constrained). We refer to these two aspects of spatial resilience as  \emph{recoverability} and \emph{persistency}, respectively. Example~\ref{example 1} illustrates spatial resiliency using a solar power-driven rover.

To reason formally about the correctness of spatial behavior, Spatio-Temporal Reach and 
Escape Logic (STREL)~\cite{2022STREL} extends Signal Temporal Logic (STL) with \emph{reach} and \emph{escape} operators.  Specifically, the \emph{reach} operator $\varphi_1 \reach_{[d_1, d_2]}^f \varphi_2$ describes the behavior of reaching a location satisfying $\varphi_2$ through a route $\tau$ such that the length of this route belongs to the interval $[d_1, d_2]$, and $\varphi_1$ is satisfied at all locations.
The \emph{escape} operator $\escape_{[d_1,d_2]}^f$ describes the possibility of escaping from a certain region via a route that passes only through locations satisfying $\varphi_1$, such that the distance between the starting location and the last location belongs to the interval $[d_1, d_2]$. STREL has been applied 
to mobile ad-hoc networks, transportation systems, and sensor networks, providing 
a powerful framework for specifying spatio-temporal properties.  STREL, however, is not designed to quantify resilience: it captures spatial reachability and 
safety but not how long (distance-wise) it takes a system to recover from a violation (recoverability) or for how long it can sustain a desired condition 
(persistency).  This limitation motivates the need for a new logic that integrates resilience into spatial reasoning.  

\begin{example}[Solar power-driven rover]
\label{example 1}
    \begin{figure}[t]
        \centering
        \includegraphics[width=0.8\linewidth]{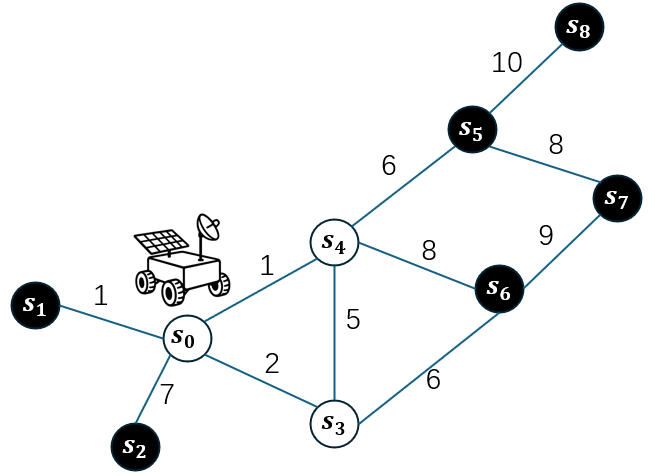}
        \caption{Illustration of a solar power-driven rover and the corresponding spatial distribution of solar power.
        Locations satisfying the property  \texttt{$\varphi= \text{solar\_power}\geq1\text{kW}$} are shown as black nodes; white nodes indicate locations where the property does not hold.} 
        \label{fig:example}
    \end{figure}

Consider a rover tasked with exploring a new region. It relies on an onboard battery as its sole energy source, which, when fully charged, allows the rover to travel up to 10~km.
The battery can be recharged using solar power, but the available solar power varies by location.
Assume the rover can only be fully recharged at locations where the solar power is at least 1 kW, as specified by the (STREL) formula \texttt{$\varphi= \text{solar\_power}\geq1\text{kW}$}. 
Figure~\ref{fig:example} illustrates the region as an A-spatial model to represent the region, where black nodes represent locations satisfying $\varphi$ and white nodes indicate locations where $\varphi$ does not hold. 
Edge weights denote the distances between nodes in kilometers. 
The rover begins at node $s_0$ with enough battery to travel 9~km. 
Since it cannot recharge at this starting location, we must evaluate its spatial resilience w.r.t.\ $\varphi$ at $s_0$ as follows:

\noindent
\emph{Recoverability}:  The rover must reach a charging node (where $\varphi$ holds) depleting its battery, i.e., traveling 7~km to $s_5$ or 8~km to $s_6$. 
Recoverability quantifies the minimum distance the rover needs to travel to a recharging node. 
For example, from $s_0$, the rover can take the route $s_0\rightarrow s_3\rightarrow s_6$, covering 8 km. 
Thus, the recoverability for this route is 8 km. 

\noindent
\emph{Persistency}. Post-recovery, the rover should be able to continue
exploring along a route comprising only charging nodes.
Persistency thus quantifies how far the rover can travel post-recovery by
visiting locations satisfying $\varphi$. 
Persistency post-recovery at $s_6$ is 27 km, since it can travel along $s_6 \rightarrow s_7 \rightarrow s_5 \rightarrow s_8$ and the distance of this route is 27~km.
\end{example}

\paragraph{Our Contributions.} In this paper, we introduce a formal framework for reasoning about spatial resiliency in CPS. 
Our approach builds on the spatial fragment of STREL, excluding temporal operators, which we refer to as SREL.
Within this framework, spatial resiliency is syntactically described using \emph{Spatial Resiliency Specifications} (SpaRS), whose atomic elements we call S-atoms.
Given any SREL formula $\varphi$ and distance parameters $d_1, d_2$, we define the concept of S-atom $S_{d_1, d_2} (\varphi)$ as the SREL formula $\neg\varphi\reach_{[0,d_1]} (\varphi \reach_{[d_2, +\infty)}\varphi)$.
This specification requires that the system can reach a location where $\varphi$ holds within a distance of $d_1$ (recoverability), and that $\varphi$ continues to hold for at least a further distance of $d_2$ (persistency). 
S-atoms can be composed using SREL operators to express more complex spatial resilience requirements such as conjunctions (e.g., $S_{d_1, d_2}(\varphi_1) \wedge S_{d_3,d_4}(\varphi_2)$), or the requirement that every location satisfies an S-atom (e.g., $\boxbox_{[d_3,d_4]} S_{d_1,d_2}(\varphi)$).
To capture the quantitative aspects of resilience, we introduce the \emph{Spatial Resilience Value} (SpaRV) function $\spresv$, which provides a set of $(rec, per)$ pairs for each S-atom.
Here, $rec$ measures how much shorter the actual recovery route is compared to the bound $d_1$, while $per$ measures how much longer the persistency route is compared to $d_2$.
We prove that our semantics for SpaRS is both sound and complete with respect to the underlying STREL semantics.

Importantly, our framework does not impose any preference or aggregation between recoverability and persistency.
As a result, the semantics may yield multiple non-dominated $(rec, per)$ pairs, reflecting different trade-offs between these two aspects. In such cases, we retain all non-dominated pairs to provide a thorough, assumption-free characterization of spatial resilience, using Pareto optimization to derive the semantics from subformulas.
For instance, in Example~\ref{example 1}, starting from $s_0$, there are distinct, non-dominated recovery-persistency paths: one recovers via $s_0 \rightarrow s_4 \rightarrow s_5$ (distance 7) and persists through $s_5 \rightarrow s_7 \rightarrow s_6$ (distance 17); another recovers $s_0 \rightarrow s_3 \rightarrow s_6 \rightarrow s_7$ (distance 8) and persists along $s_6 \rightarrow s_7 \rightarrow s_5 \rightarrow s_8$ (distance 27). 
The former achieves better recoverability (shorter recovery), while the latter yields better persistency (longer sustained satisfaction).

We also present an exact algorithm for computing the SpaRV of a SpaRS formula. For the S-atoms, we use the Dijkstra's algorithm to calculate the recoverability and depth-first search to calculate the persistency. For the Reach and Escape operators, we adopt a flooding-based algorithm.
To demonstrate the practicality and flexibility of our approach, we conduct comprehensive experiments on two case studies: networked microgrids and an urban bike-sharing system. 
In both scenarios, we formalize key operational requirements in SREL and evaluate their SpaRVs under various SpaRS specifications. 
Our findings highlight the expressive power and versatility of our spatial resilience framework.


\section{Preliminaries}

In this section, we introduce the A-spatial model and a fragment of the Spatial-Temporal  Reach and Escape Logic~\cite{2022STREL}.

\subsection{A-Spatial Model}
We adopt the definition of the A-spatial model from Nenzi et al.~\cite{2022STREL}. An A-spatial model is a graph with edges having weights from a set $A$.
We consider undirected weighted graphs.
\begin{definition}[A-spatial model~\cite{2022STREL}] 
    An \emph{A-spatial model} $\lserv$ is a pair $(L, \wfun)$ where:
    \begin{itemize}
        \item $L$ is a finite set of \emph{locations}; $L$ is sometimes referred to as the \emph{spatial universe};
        \item $\wfun\subseteq L\times A\times L$ is a \emph{proximity function} associating at most one label $w \in A$ with each distinct pair $l_1, l_2\in L$. 
    \end{itemize} 
\end{definition}

We will equivalently write  $(l_1,w,l_2)\in W$ as $W(l_1,l_2)=w$ or $\nextto{l_1}{w}{l_2}$, meaning that $l_1$ is \emph{next to} $l_2$ with weight $w \in A$.

\begin{definition}[Route~\cite{2022STREL}]  
Let $\lserv=\langle L,\wfun \rangle$ be an A-spatial model.  A \emph{route} $\route$ through $\lserv$
is a sequence $l_0 l_1\cdots l_k$ such that $l_i \in L$ and $(l_i, w, l_{i+1}) \in W$ and no edge $(l_i, w, l_{i+1})$ occurs more than once (but repeated locations are allowed). A \emph{simple route} is a route where no location repeats.
\end{definition}

Let $\route=l_0 l_1\cdots l_k$ be a route, $i\in \mathbb{N}$ and $l_i \in L$. We adopt the following notation:
\begin{itemize}
\item $\route[i]$ denotes the $i$-th node in the route, namely $l_i$.
\item $\route[..i]$ denotes the prefix of the route up to and including the i-th node, i.e., $\route=l_0 l_1\cdots l_i$.
\item $\tau[i..]$ denotes the suffix of the route starting from the $i$-th node, i.e., $l_i l_{i+1} \cdots l_k$.
\item We write $l \in \route$ if there exists an index $i$ such that $\route[i]=l$.
\item Conversely, we write $l\not\in \route$ if no such index exists, that is, if the label $l$ does not appear in the route.

\end{itemize}


Note that since the spatial universe $L$ is finite and there are no repeated edges in a route, the number of routes in a spatial model is finite, and the number of edges in a route $\route$ (denoted $|\route|$) is also finite. We use $Routes(\lserv)$ to denote the set of routes in $\lserv$, while $Routes(\lserv,l)$ denotes the set of routes starting from $l \in L$.

\begin{definition}[Distance Domain~\cite{2022STREL}]
\label{def:distDom}
We define a distance domain as a structure  $(D,\bot_D,\top_D,+_{D},\leq_{D})$ where:
\begin{itemize}
    \item $D$ is a set equipped with a total order $\leq_{D}$.
    \item $\bot_{D}$ is the least element (minimum) of $D$ under $\leq_{D}$.
    \item $\top_{D}$ is the greatest  element (maximum) of $D$ under $\leq_{D}$.
    \item $(D,\bot_D,+_{D})$ forms a monoid, i.e., $+_{D}$ is an associative binary operation with identity $\bot_D$.
\end{itemize}
\end{definition}

\begin{definition}[Distance functions and Route Length~\cite{2022STREL}]
\label{def:distance}
  Let \(\lserv = (L, W)\) be an \(A\)-spatial model, and \((D, \bot_D, \top_D, +_D,
  \leq_D)\) be a distance domain.
  A \emph{distance function} \(f: A \to D\) is a mapping that assigns to each edge weight in $A$ a value in the distance domain. Given the distance function $f$, the \emph{route length}
  \(d^f: Routes(\lserv) \to D\) is defined as:

  \begin{displaymath}
    d^f(\tau) =
    \begin{cases}
      \sum_{i=1}^{|{\tau}|} f(W(\tau[i-1], \tau[i])) & ~\text{if}~ |\tau| \geq 1 \\
      \bot_D                                             & ~\text{otherwise},
    \end{cases}
  \end{displaymath}
  where the summation is defined over the monoid \(+_D\). In this paper, we use the distance domain
$ (\mathbb{R}_{\ge 0}\cup\{+\infty\},\ 0,\ +\infty,\ +,\ \le)$.

\end{definition}

We use the notation \(d^f_\tau[i]\) to refer to \(d^f(\tau[..i])\), i.e., the
distance of the prefix of a route \(\tau\) up to the \(i\)th location on
the route, and we use $d_{\tau}^f[i..]$ to refer to $d_{\tau}^f([i..])$. We use~0 for the index $i$ of the first node in a route. 
We define the distance between two locations $l_1$ and $l_2$ in $\lserv$ to be the length of the shortest route connecting the two locations:
\begin{displaymath}
  d^f_{\lserv}[l_1,l_2] = \min\Set{d^f_\tau[i] \given \tau \in Routes(\lserv,l_1),
    \tau[i] = l_2}.
\end{displaymath}

\begin{definition}[Spatial signal] Let $L$ be a \emph{spatial universe}, and $\mathbb{R}$ be the set of real numbers. A \emph{spatial signal} is a function $\ssign: L\rightarrow \mathbb{R}$.  
\end{definition}

\subsection{Spatial Reach and Escape Logic (SREL)}
\label{SEC:SREL}

STREL~\cite{2022STREL} is a logic for specifying the spatial and temporal behavior of a spatio-temporal signal over an A-spatial model. In this paper, we consider static A-spatial models and spatial signals that do not change over time. Thus, we consider the \emph{spatial} fragment of STREL as our basis, which we call SREL.

\begin{definition}[SREL syntax~\cite{2022STREL}]
    \begin{equation*}
    \varphi := 
    \mu \mid \neg \varphi \mid \varphi_1 \land \varphi_2 
    \mid \varphi_1 \reach_{[d_1, d_2]}^{f} \varphi_2
    \mid \escape_{[d_1, d_2]}^{f} \varphi
  \end{equation*}
\end{definition}

Here, \(\mu\) is an atomic predicate, negation \(\neg\), conjunction \(\land\), and disjunction \(\lor\) are the standard Boolean connectives. 
\emph{Reach} $\reach$ and \emph{Escape} $\escape$ are spatial operators, where $d_1, d_2 \in D$ represents distances in the distance domain $D$, and $f$ is a distance function mapping to $D$.
The \emph{Reach} operator \(\varphi_1 \reach_{[d_1,d_2]}^f \varphi_2\) specifies reaching a location satisfying $\varphi_2$ through a route $\tau$ such that $d_1 \leq d^f(\tau) \leq d_2$ and $\varphi_1$ is satisfied at all locations along the route.
The \emph{Escape} operator \(\escape_{[d_1, d_2]}^f \varphi\) specifies escaping from a certain region via a route that passes only through locations that satisfy \(\varphi\) and the distance between the starting location $l$ and the last location $l'$ satisfies $d_1\leq d_{\lserv}^f[l,l']\leq d_2$.
Note that for \emph{Escape}, the distance constraint is not necessarily over the satisfying route, but the shortest route between the start and end locations. For \emph{Reach}, the distance constraint is with respect to the length of the satisfying route. 

We also consider two other spatial operators, the Somewhere
$\diamonddiamond$ and Everywhere $\boxbox$ operators, in our case studies:
\begin{equation*}
    \diamonddiamond_{[0, d]}^f \varphi = \top \Reach_{[0, d]}^f \varphi , \quad
    \boxbox_{[0, d]]}^f \varphi = \neg \diamonddiamond_{[0, d]]}^f \neg
        \varphi
\end{equation*}
Somewhere operator $\diamonddiamond_{[0, d]}^f \varphi$ specifies, from a certain origin, a reachable destination that satisfies $\varphi$ through a route whose length belongs to the interval $[0,d]$. 
Everywhere operator $\boxbox_{[0, d]}^f \varphi$ specifies that all routes, starting from a certain location, of length within $[0, d]$ lead to locations where $\varphi$ holds.

\begin{definition}[SREL qualitative and quantitative  semantics~\cite{2022STREL}]
\label{def: srel semantics}
Let $\lserv=(L,W)$ be an $A$-spatial model, $\xi$ an spatial signal over $L$,  $l\in L$, and $\varphi$ an SREL formula.  
Function $\lambda$ maps $\xi(l)$ to a value in $\{\top, \bot\}$, associated with the atomic predicate $\mu$.
The satisfaction relation $(\lserv, \xi, l) \models \varphi$, indicating that signal $\xi$ satisfies $\varphi$ at location $l$, is defined as follows:

\noindent
{\footnotesize
\begin{minipage}[b]{1.\linewidth}

\vspace{0.1cm}
\begin{tabular}[b]{l@{\hspace{2.5pt}}c@{\hspace{2.5pt}}l}

$(\lserv, \xi, l)\models \mu$ & $\!\Leftrightarrow$\! & $\lambda(\xi(l))= \top$ \\[.2cm]

$(\lserv, \xi,  l)\models \neg \varphi$ & $\!\Leftrightarrow\!$ & $\neg((\lserv, \xi, l) \models \varphi)$ \\[.2cm]

$(\lserv, \xi, l) \models \varphi_1 \wedge \varphi_2$ & $\!\Leftrightarrow\!$ & $(\lserv, \xi, l)\models \varphi_1 \wedge (\lserv, \xi, l)\models \varphi_2 $  \\[.2cm]  

$(\lserv, \xi, l)\models \varphi_1 \reach_{[d_1,d_2]}^f \varphi_2$ & $\!\Leftrightarrow\!$ & $\exists \tau\in Routes(\lserv, l), i\in\mathbb{N}  \; \text{s.t.}\; d_\tau^f[i]\in[d_1, d_2]$ \\
&&$ \wedge (\lserv, \xi, \tau[i]) \models \varphi_2 \wedge \forall j<i, (\lserv, \xi, \tau[j])\models \varphi_1$ \\[.2cm] 
 
$(\lserv, \xi, l)\models \escape_{[d_1,d_2]}^f \varphi$ & $\!\Leftrightarrow\!$ & $\exists \tau\in Routes(\lserv, l), i\in\mathbb{N} \;\text{s.t.}\; d_\lserv^f[l, \tau[i]]\in [d_1, d_2]$ \\

&&$\wedge \forall j\leq i, (\lserv, \xi, \tau[j]) \models \varphi$\\

\end{tabular}

\end{minipage}
	\label{tab:STREL semantics}
}
\vspace{-0.2cm}

The quantitative semantics $\mathbf{n}(\lserv, \xi, \varphi, l)$ is recursively defined as follows. 
Function $v$ maps $\xi(l)$ to a value in $\mathbb{R}$, associated with $\mu$.

\noindent
{\footnotesize
\begin{minipage}[b]{1.\linewidth}
\vspace{0.1cm}
\begin{tabular}[b]{rcl}


$\mathbf{n}( \lserv, \xi, \mu, l)$ & $=$ & $v(\xi(l))$ \\[.2cm]

$\mathbf{n}(\lserv, \xi, \neg\varphi, l)$ & $=$ & $ -\mathbf{n}(\lserv, \xi, \varphi, l)$  \\[.2cm]

$\mathbf{n}( \lserv, \xi, \varphi_1 \wedge \varphi_2, l)$ & $=$ & $ \min(\mathbf{n}( \lserv, \xi, \varphi_1, l) , \mathbf{n}(\lserv, \xi, \varphi_2, l) )$  \\[.2cm]  

      

 $\mathbf{n}(\lserv, \xi, \varphi_{1} \: $ $\reach_{[d_{1},d_{2}]}^f \: \varphi_{2}, l)$ & $=$ \\ 
 
\multicolumn{3}{r}{
$\max\limits_{\tau\in Routes(\lserv,l)}
        ~~\max\limits_{i : \left(d_{\tau}^{f}[i] \in [d_{1},d_{2}]\right)}
        \left(
            \min (\mathbf{n}(  \lserv, \xi, \varphi_2,  \tau[i]),
            
            \min\limits_{j < i} 
                \mathbf{n}(  \lserv, \xi, \varphi_1,  \tau[j]) )        %
        \right)$} \\[.2cm]

$\mathbf{n}( \lserv, \xi,$ $\escape_{[d_{1},d_{2}]}^{f} \: \varphi, l)$ & $=$ & \\
\multicolumn{3}{r}{
$\max\limits_{\tau\in Routes(\lserv,l)}
        ~~\max\limits_{i: \left(d_{\lserv}^{f}[l,\tau[i]] \in [d_{1},d_{2}]\right)}
            ~~\min\limits_{j \leq i} 
                \mathbf{n}(  \lserv, \xi, \varphi,  \tau[j])$} \\
\end{tabular}
\end{minipage}
	\label{tab:STREL semantics}
}

\end{definition}

\section{Specifying  Spatial Resilience}
In this section, we introduce our Spatial Resilience Specification (SpaRS) language and its quantitative semantics in terms of non-dominated recoverability-persistency pairs. 


\subsection{Spatial~Resilience~Specification~(SpaRS) Language }
We introduce SpaRS to reason about the resilience of SREL formulas. Given an SREL formula $\varphi$, the spatial resiliency should capture both recoverability and persistency. For this purpose, we use 
$\sr_{d_1, d_2} (\varphi) \equiv \neg \varphi \reach_{[0, d_1]} (\varphi \reach_{[d_2, \infty)} \varphi)$ 
to formalize this notion and we will call it an S-atom. 
An S-atom describes the requirement if $\varphi$ is violated at the current node, it can be recovered at another node within distance $d_1$; also, it can be subsequently maintained on a route of length at least $d_2$. 
We give the syntax of our SpaRS logic as follows.

\begin{definition}[Spatial Resiliency Specification (SpaRS)]

    \begin{equation*}
    \psi := 
     \sr_{d_1, d_2} (\varphi) \mid 
     \neg \psi \mid \psi \land \psi 
    \mid \psi \reach_{[d_1, d_2]}^{f} \psi
    \mid \escape_{[d_1, d_2]}^{f} \psi
  \end{equation*}
  where $\varphi$ is an SREL formula, $\sr_{d_1, d_2} (\varphi) \equiv \neg \varphi \reach_{[0, d_1]} (\varphi \reach_{[d_2, \infty)} \varphi)$, $d_1, d_2 \in \mathbb{R}_{\geq 0}$, and $d_1 \leq d_2$. 
  
\end{definition}

\begin{remark}
    Why do we define $\sr_{d_1, d_2} (\varphi) = \neg \varphi \reach_{[0,d_1]} (\varphi \reach_{[d_2,\infty)} \varphi)$ instead of $\sr_{d_1, d_2} (\varphi) = \neg \varphi \reach_{[0, d_1]} \escape_{[d_2, \infty)} \varphi$ as the atomic proposition of SpaRS?
    The main reason is that the operators $\reach$ and $\escape$ use different distance metrics in their subscripts. 
    The $\reach$ operator uses the ``route length'' while the $\escape$ operator uses the ``distance between two locations'' as described in Definition~\ref{def:distance}. For example, consider Figure~\ref{fig:example} and two formulas $\neg \varphi \reach_{[0,9]} (\varphi \reach_{[25,\infty)} \varphi)$ and $ \neg \varphi \reach_{[0, 9]} \escape_{[25, \infty)} \varphi$ evaluated at location $s_0$. The former expression is satisfied because it can recover to $s_6$ through route $s_0\rightarrow s_3 \rightarrow s_6$ of length $8<9$, then persist along route $s_6\rightarrow s_7\rightarrow s_5\rightarrow s_8$ of length $27>25$. The latter expression is violated since we cannot find a route post-recovery such that the distance $d^f_{\lserv}[l_1,l_2]$ between the first location $l_1$ and the last location $l_2$ of this route is at least 25. The most promising route is $s_6\rightarrow s_7\rightarrow s_5\rightarrow s_8$; however, $d^f_{\lserv}[s_6,s_8] = 8+6+10=24<25$.

\end{remark}

We next provide a quantitative semantics for SpaRS specifications in the form of \emph{Spatial Resilience Value (SpaRV)}. Intuitively, SpaRV quantifies the extent to which  recoverability and persistency are satisfied. More precisely, it produces a non-dominated set of pairs $(x_r, x_p)\in \mathbb{Z}^2$, where (in the atomic case) $x_r$ quantifies how close recovery occurs compared with distance $d_1$, and $x_p$ quantifies, post-recovery, how far the property is maintained compared with distance $d_2$. We further demonstrate the soundness and completeness of the SpaRV-based semantics w.r.t. the SREL Boolean interpretation of spatial resiliency specifications. 
We first define methods for comparing two recoverability-persistency pairs.

Following ~\cite{chen2022resilience}, we adopt the notion of ``resilience dominance'' captured by the relation $\succre$ in $\mathbb{Z}^2$. This is needed because using the standard Pareto-dominance relation $\succ$ would result in an ordering of SpaRV pairs that is inconsistent with the Boolean satisfiability viewpoint. Consider the pairs $(-3,2)$ and $(1,1)$. By Pareto-dominance, $(-3, 2)$ and $(1,1)$ are mutually non-dominated, but a SpaRV of $(-3,2)$ indicates that the system does not satisfy recoverability; namely, the recovery occurs 3 units farther than the bound. On the other hand, a SpaRV of $(1,1)$ implies satisfaction of both recoverability and persistency bounds, and thus should be preferred to $(-3,2)$. We formalize this intuition as follows.

\begin{definition}[Resiliency Binary Relations~\cite{chen2022resilience}]\label{def:binaryrelation}
We define binary relations $\succ_{re}$, $=_{re}$, and $\prec_{re}$ in $\mathbb{Z}^2$.  Let $x,y\in \mathbb{Z}^2$ with $x=(x_r,x_p)$, $y=(y_r,y_d)$, and \emph{sign} is the signum function. 
We have that $x \succ_{re} y$ if one of the following holds:
\begin{enumerate}
    \item $sign(x_r)+sign(x_p) = sign(y_r)+sign(y_d)$, and $x \succ y$.
    \item $sign(x_r)+sign(x_p) > sign(y_r)+sign(y_d)$.
\end{enumerate}
We denote by $\prec_{re}$ the dual of $\succ_{re}$. If neither $x\succ_{re} y$ nor $y \prec_{re} x$,\footnote{This is equivalent to stating that $sign(x_r)+sign(x_p) = sign(y_r)+sign(y_d)$, and that neither $x\succ y$ nor $y \succ x$.} then $x$ and $y$ are \emph{mutually non-dominated}, denoted $x =_{re} y$.
Under this ordering, a \emph{non-dominated set} $S$ is such that $x=_{re} y$ for all $x,y \in S$. 
\end{definition}


\begin{definition}[Maximum and Minimum Resilience Sets~\cite{chen2022resilience}]
    Given $P\subseteq \mathbb{R}^2$ with $P\not=\emptyset$, the maximum resilience set of $P$, denoted $\maxre(P)$, is the largest subset of $P$ such that $\forall x\in \maxre(P)$ and $\forall y\in P$, we have $x\succre y$ or $x\eqre y$. 
    The minimum resilience set of $P$, denoted $\minre(P)$, is the largest subset of $P$ such that $\forall x\in \maxre(P), \forall y\in P$, such that $x\precre y$ or $x\eqre y$.
\end{definition}

Next, we formally introduce the semantics for our SpaRS logic. Its definition makes use of maximum and minimum resilience sets. Our semantics produces non-dominated sets, which implies that all pairs in such a set are equivalent from a Boolean satisfiability standpoint. This is because $x_r>0$ ($x_p>0$) implies Boolean satisfaction of the recoverability (persistency) portion of an $S_{d_1, d_2} (\varphi)$ expression. This property will be useful in Theorem ~\ref{theorem: soundness and completeness}, where we show that our semantics is sound with respect to the Boolean semantics of SREL.

\begin{definition}[Spatial Resilience Value (SpaRV)] 
\label{sprs quantitative semantics}
Let $\lserv=(L,W)$ be an $A$-spatial model, $\xi$ an spatial signal over $L$, $l\in L$, and $\psi$ a SpaRS specification.
The SpaRV $\spresv(\lserv, \xi, \psi, l)$ of $\psi$ with respect to $\xi$ at $l$ is given as a set of distance pairs $(x_r, x_p)\in\mathbb{R}^2$. 

For $\psi$ an S-atom of the form $\sr_{d_1,d_2}(\varphi), \varphi$ an SREL formula, 

\noindent
{\footnotesize
\begin{align}
\spresv(\lserv, \xi, \sr_{d_1, d_2} (\varphi), l)
&= \maxre_{\tau\in Routes(\lserv, l)}
\{(-\infty,-\infty)\} \cup \big\{(x_r, x_p) ~\mid~ \exists i \;\text{s.t.} \nonumber\\
\multicolumn{2}{r}{
$\qquad
\forall j<i,\; (\lserv, \xi, \tau[j]) \models \neg \varphi
\;\land\;
\forall k\geq i,\; (\lserv, \xi, \tau[k]) \models \varphi
\big\}$
}
\label{equ:sprs-atom-semantics}
\end{align}
}

where $x_r=d_1-d_{\tau}^f[i]$, and $x_p=d_{\tau}^f[i..]-d_2$.


The SpaRV of a composite SpaRS formula is defined inductively as follows.
\vspace{-0.2cm}
    \begin{table}[!h]
    \hspace*{-0.4cm}
{\footnotesize
\begin{minipage}[b]{1.\linewidth}

\begin{tabular}[b]{rcl}


$\spresv(\lserv, \xi, \neg\psi, l)$ & $=$ & $ \{(-x,-y): (x,y) \in \spresv(\lserv, \xi, \psi, l) \}$  \\[.1cm]

$\spresv( \lserv, \xi, \psi_1 \wedge \psi_2, l)$ & $=$ & $\minre (\spresv( \lserv, \xi, \psi_1, l) \cup  \spresv(\lserv, \xi, \psi_2, l))$  \\[.1cm]

$\spresv( \lserv, \xi, \psi_1 \vee \psi_2, l)$ & $=$ & $\maxre (\spresv( \lserv, \xi, \psi_1, l) \cup  \spresv(\lserv, \xi, \psi_2, l))$  \\[.1cm]

 $\spresv(\lserv, \xi, \psi_{1} \: \reach_{[d_{1},d_{2}]}^f \: \psi_{2}, l)$ & $=$ \\[-0.1cm]
 
\multicolumn{3}{r}{

$\maxre\limits_{\tau\in Routes(\lserv,l)}
        ~~\maxre\limits_{i : \left(d_{\tau}^{f}[i] \in [d_{1},d_{2}]\right)}
    \minre\left(
            (\spresv(  \lserv, \xi, \psi_2,  \tau[i]),
            
            \minre\limits_{j<i}(
                \spresv(  \lserv, \xi, \psi_1,  \tau[j]) )         %
        \right)$} \\[-0.05cm]
        
$\spresv( \lserv, \xi,\escape_{[d_{1},d_{2}]}^{f} \: \psi, l)$ & $=$ & \\

\multicolumn{3}{r}{

$\maxre\limits_{\tau\in Routes(\lserv,l)}
        ~~\maxre\limits_{i : \left(d_{\lserv}^{f}[l,\tau[i]] \in [d_{1},d_{2}]\right)}
            ~~\minre\limits_{j \leq i} 
                (\spresv(  \lserv, \xi, \psi,  \tau[j]) )$ \hspace{.3cm} } \\

$\spresv(\lserv, \xi, \diamonddiamond_{[d_1, d_2]} \psi)$ &=& $\maxre\limits_{\tau \in Routes(\lserv, l)} \maxre\limits_{i : \left(d_{\tau}^{f}[i] \in [d_{1},d_{2}]\right)} \spresv(\lserv, \xi, \psi, \tau[i])$  \\

$\spresv(\lserv, \xi, \boxbox_{[d_1,d_2]} \psi) $ &= & $\minre\limits_{\tau \in Routes(\lserv, l)} \minre\limits_{i : \left(d_{\tau}^{f}[i] \in [d_{1},d_{2}]\right)} \spresv(\lserv, \xi, \psi, \tau[i])$ 
\end{tabular}

\end{minipage}
	\label{tab:SpREL quantitative Semantics}
	}
\end{table}

\end{definition}


We first look at Eq.~(\ref{equ:sprs-atom-semantics}). Suppose there are multiple routes starting from $l$ with the property that 
there exists an index $i$ such that all nodes with index $<i$ satisfy $\neg \varphi$ and all nodes with index $\geq i$ satisfy $\varphi$.
In this case, $\tau[i]$ is the first node at which $\varphi$ becomes satisfied over $\tau$ starting from $l$, and afterwards $\varphi$ remains satisfied at all locations of $\tau$, corresponding to recoverability and persistency, respectively. 
Thus, we use $d_{\tau}^f[i]$, the distance over $\tau[..i]$, to quantify recoverability and $d_{\tau}^f[i..]$, the distance over $\tau[i..]$, to quantify persistency. 
Value $x_r=d_1-d_{\tau}^f[i]$ quantifies how close the recoverability node occurs compared with the distance bound $d_1$. If $x_r>0$, the recoverability bound requirement is satisfied because a shorter (closer) recovery implies satisfaction. $x_p=d_{\tau}^f[i..] -d_2$ quantifies how long $\varphi$ is maintained compared with the distance bound $d_2$. If $x_p>0$, the persistency bound requirement is satisfied because longer persistency implies satisfaction. Furthermore, $\spresv(\lserv, \xi, S_{d_1,d_2}(\varphi), l)$ is a maximum resilience set that includes maximum $(x_r, x_p)$ pairs corresponding to different routes. If there does not exist a route $\tau$ starting from $l$ that reaches a location at which $\varphi$ is satisfied, then the semantics return $(-\infty, -\infty)$, representing the worst-case recoverability and persistency.

The SpaRS semantics for composite SpaRS formulas is derived by computing sets of maximum/minimum $(x_r, x_p)$ pairs. For $\land$ and $\boxbox$, we consider the minimum resilience set. For $\lor$ and $\diamonddiamond$, we consider the maximum resilience set. For $\reach$ and $\escape$, we consider nested maximum and minimum resilience sets. 

\begin{remark}
In Eq.~(\ref{equ:sprs-atom-semantics}),
we take an existential view of recoverability and persistency in terms of the branching of the underlying A-spatial model. That is, the same route used for recovery is extended for persistency. An alternative approach involves a universal perspective of persistency, taking into account all routes emanating from the recovery location.
Our future work will examine this issue in more detail.
\end{remark}

\begin{example}
    In Example ~\ref{example 1}, $\varphi$ refers to whether the solar power at a location is sufficient to fully charge the rover. We evaluate the semantics of the following SpaRS formulas at location $s_0$ of Figure~\ref{fig:example}.  
    \begin{itemize}[leftmargin=*]
        \item $\sr_{9, 12}(\varphi)$.
        This semantic gives two non-dominated pairs. One route $s_0 \rightarrow s_3 \rightarrow s_6 \rightarrow s_7\rightarrow s_5\rightarrow s_8$ induces the pair $(x_r, x_p) = (9-8, 27-12) = (1,15)$. Another route $s_0 \rightarrow s_4 \rightarrow s_5\rightarrow s_7 \rightarrow s_6$, induces the pair $(x_r, x_p) = (9-7, 17-12)= (2,5)$. Since $(1,15)$ and $(2, 5)$ are non-dominated, we have $ \spresv(\lserv, \xi, \sr_{14, 6} (\varphi), s_0) = \{(1,15),\allowbreak (2,5)\}$. 
        
        
        \item $\boxbox_{[0,5]} \sr_{9,12} (\varphi)$. Locations $s_0, s_1, s_3, s_4$ are within the distance $[0,5]$ of $s_0$. After calculation, We know $\spresv(\lserv, \xi, \sr_{9, 12} (\varphi), s_0) $ $= $ $\{(2,5), \allowbreak (1,15)\} $, 
        $
        \spresv(\lserv, \xi,$ $ \sr_{9, 12} (\varphi), s_1)$ $ = $ $\{(9,-12)\},$
        $        \spresv(\lserv, \xi,$ $ \sr_{9, 12} (\varphi), s_3)$ $ = $ $\{(3,15)\}$, 
        $        \spresv(\lserv, \xi, \sr_{9, 12} (\varphi), s_4) $ $= $ $\{(3,$$ 5),$ $ (1,15)\} 
        $. The everywhere operator selects the minimum pairs from them. Thus, $\spresv(\lserv, \xi, \allowbreak \boxbox_{[0,5]} \sr_{9,12} (\varphi), \allowbreak s_0) = \{(9,-12)\}$.     
    \end{itemize}
    
\end{example}

\begin{remark}
We define the S-atom SpaRV over routes without repeated edges for two main reasons.
First, allowing edges to repeat would make the persistency distance arbitrarily large or even unbounded, which is not meaningful for resilience analysis. 
For example, if two adjacent nodes $l_1$ and $l_2$ both satisfy $\varphi$, then the route $l_1 \rightarrow l_2 \rightarrow l_1 \rightarrow l_2 \rightarrow \cdots$ can be extended infinitely, yielding infinite persistency. 
Second, recoverability is determined by the shortest distance to a node satisfying $\varphi$, so any route that repeats edges is necessarily longer and therefore irrelevant for computing recovery distance. 
Following the convention in STREL, we still allow repeated nodes (provided that no edge is repeated), which preserves expressive flexibility without introducing routes with infinite length.
\end{remark}

\begin{remark}
We choose STREL as the foundation for SpaRS because it is the only existing spatial logic that provides route-based operators that align directly with the notions of recoverability and persistency. Other spatial logics differ fundamentally that make them unsuitable for defining resilience semantics.
(1) SSTL~\cite{SSTL} resembles STREL syntactically, but its spatial operators---somewhere and surround---are region-oriented rather than path-oriented. 
The somewhere operator can express the existence of a recovery point but cannot encode persistency along a spatial route where all locations are satisfied, while the surround operator describes topological enclosure, which is not meaningful for our route-based semantics.
(2) SpaTeL~\cite{SpaTel} relies on quad-tree spatial abstractions and directional spatial modalities. 
Such hierarchical spatial structures do not generalize to arbitrary spatial domains, and SpaTeL formulas are often difficult to specify manually, typically requiring automated synthesis or learning.
(3) SaSTL~\cite{SaSTL} extends STL with spatial aggregation and counting operators suited for expressing global or statistical spatial properties, but its semantics do not refer to spatial routes and therefore cannot characterize resilience behaviors that depend on evolution along a route.
\end{remark}

\begin{theorem}[Soundness and Completeness of SpaRS Semantics]
\label{theorem: soundness and completeness}
Let $\xi$ be a spatial signal and $\psi$ a SpaRS specification. The following results at location $l$ hold:

(1) $\exists x \in \spresv(\lserv, \psi, \xi, l) \; s.t. \; x\succre \mathbf{0} \Longrightarrow(\lserv, \xi, l) \models \psi $

(2) $\exists x \in \spresv(\lserv, \psi, \xi, l) \; s.t. \; x\precre \mathbf{0} \Longrightarrow(\lserv, \xi, l) \models \neg \psi $

(3) $(\lserv, \xi, l) \models \psi \Longrightarrow \exists x \in \spresv(\lserv, \psi, \xi, l) \; s.t. \; x\succre \mathbf{0} \; or \;  x\eqre \mathbf{0}$

(4) $(\lserv, \xi, l) \models \neg \psi \Longrightarrow \exists x \in \spresv(\lserv, \psi, \xi, l) \; s.t. \; x\precre \mathbf{0} \; or \;  x\eqre \mathbf{0}$
    
\end{theorem}

\begin{proof}
We denote $x=(x_r, x_p)$ and prove the theorem inductively.
We prove statements (1) and (3), while statements (2) and (4) follow analogously. 
Due to space limitations, we do not show proof for the somewhere and everywhere operators, which can be derived from the reach operator.

\vspace{1ex}
\noindent (1)~$\exists x \in \spresv(\lserv, \psi, \xi, l) \; s.t. \; x\succre \mathbf{0} \Longrightarrow(\lserv, \xi, l) \models \psi$

\noindent\textbf{Case} $\psi=\sr_{d_1, d_2} (\varphi)$.
$x \succre \mathbf{0}$ implies $x_r, x_p \geq 0 \wedge x \not= \mathbf{0}$. $x_r \geq 0 $ implies 
there exists a route $\tau\in Routes(\lserv, l)$ and a recovery point $i$ such that $ d_{\tau}^f[i]\in[0,d_1]$ and $(\lserv, \xi, \tau[j]) \models \neg \varphi$ hold for all location indices $j<i$ and $(\lserv, \xi, \tau[k]) \models \varphi$ for all location indices $k\geq i$. Moreover, $x_p\geq 0$ implies that there exists a position $i'$ such that $d_{\tau'}^f[i']\in [d_2,\infty)$ and $(\lserv, \xi, \tau'[j']) \models \varphi$ hold for all positions $j'\leq i'$, where $\tau'=\tau[i..]$ is the suffix of $\tau$ starting at the recovery location. By Definition ~\ref{def: srel semantics}, $(\lserv, \xi, \tau[i]) \models \varphi \reach_{[d_2, \infty)} \varphi$, and $(\lserv, \xi, l) \models \neg\varphi \reach_{[0, d_1]} (\varphi \reach_{[d_2, \infty)} \varphi)$, where the route $\tau$ and index $i$ are witnesses to the satisfaction of the outer $\reach_{[0, d_1]}$ operator, and $\tau'$ and $i'$ are witnesses for the inner $\reach_{[d_2, \infty)}$ operator.


\noindent\textbf{Case} $\psi=\neg \psi_1$.
By Definition ~\ref{sprs quantitative semantics}, $\spresv(\lserv, \xi, \psi, l)$  $=$ $ \{(-x_r,-x_p): (x_r,x_p) \allowbreak \in \allowbreak \spresv(\lserv, \xi, \psi_1, l) \}$. Thus $\exists x \in \spresv(\lserv, \xi, \psi_1, l)$ s.t. $x \precre 0$. From the induction hypothesis, we have $(\lserv, \xi, l) \models \neg \psi_1$.  

\noindent\textbf{Case} $\psi = \psi_1 \wedge \psi_2$.
By Definition~\ref{sprs quantitative semantics}, $\spresv( \lserv, \xi, \psi, l)$  $=$  $\minre (\spresv( \lserv, \xi, \psi_1, l) \allowbreak \cup  \spresv(\lserv, \xi, \psi_2, l))$. It implies $\forall x' \in \spresv(\lserv, \xi, \psi_1, l) \cup \spresv(\lserv, \xi, \psi_2, l), x' \succre \mathbf{0}$. From the induction hypothesis, we have $(\lserv, \xi, l) \models \psi_1$ and $(\lserv, \xi, l) \models \psi_2$. Thus $(\lserv, \xi, l) \models \psi_1 \wedge \psi_2$.


\noindent\textbf{Case} $\psi = \psi_1 \vee \psi_2$.
By Definition~\ref{sprs quantitative semantics}, $\spresv( \lserv, \xi, \psi, l)$  $=$  $\maxre (\spresv( \lserv, \xi, \psi_1, l) \allowbreak \cup  \spresv(\lserv, \xi, \psi_2, l))$. It implies $\exists x' \in \spresv(\lserv, \xi, \psi_1, l) \cup \spresv(\lserv, \xi, \psi_2, l), x' \succre \mathbf{0}$. Therefore $x'\in \spresv(\lserv, \xi, \psi_1)$ or $x'\in \spresv(\lserv, \xi, \psi_2)$. From the induction hypothesis, we have $(\lserv, \xi, l) \models \psi_1$ or $(\lserv, \xi, l) \models \psi_2$. Thus $(\lserv, \xi, l) \models \psi_1 \vee \psi_2$.

\noindent\textbf{Case} $\psi = \psi_{1} \: \reach_{[d_{1},d_{2}]}^f \: \psi_{2}$.
By Definition ~\ref{sprs quantitative semantics},  $\spresv(\lserv, \xi, \psi_{1} \: \reach_{[d_{1},d_{2}]}^f \: \psi_{2}, l)= \maxre\limits_{\tau\in Routes(\lserv,l)}
        \maxre\limits_{i : \left(d_{\tau}^{f}[i] \in [d_{1},d_{2}]\right)}
    $ $\minre(
            (\spresv(  \lserv, \xi, \psi_2,  \tau[i]),            
            \minre\limits_{j<i}(
                \spresv(  \lserv, \xi, $ \, $ \psi_1, \allowbreak  \tau[j]) )            
        )$. 
This implies $\exists \tau\in Routes(\lserv, l), i\in \mathbb{N}
$, s.t. $d_{\tau}^f[i]\in[d_1,d_2]$ and $\forall x' \in \spresv(\lserv, \xi, \psi_2, \tau[i]) \, \cup \,\minre\limits_{j<i}(\spresv(\lserv, \xi, \psi_1, \tau[j])), x' \succre \mathbf{0}$. 
From the induction hypothesis, we have $(\lserv, \xi, \tau[i]) \models \psi_2$. Similarly, we have $x''\succre \mathbf{0}$ for all $x'' \in \spresv(\lserv, \xi, \psi_1, \tau[j])$ and all $j<i$.  From the induction hypothesis, $\forall j<i, (\lserv, \xi, \tau[j]) \models \psi_1$. Together with $(\lserv, \xi, \tau[i]) \models \psi_2$, we conclude $(\lserv, \xi, l )\models \psi_1 \reach_{[d_1, d_2]}^f \psi_2$.

\noindent\textbf{Case} $\psi = \escape_{[d_1,d_2]}^f \psi_1$.
By Definition~\ref{sprs quantitative semantics}, $    \spresv( \lserv, \xi,\escape_{[d_{1},d_{2}]}^{f}  \psi_1, l) =  \allowbreak$ \, $ \allowbreak
\maxre\limits_{\tau\in Routes(\lserv,l)} 
        \maxre\limits_{i : \left(d_{\lserv}^{f}[l,\tau[i]] \in [d_{1},d_{2}]\right)}
            ~~\minre\limits_{j \leq i} 
                (\spresv(  \lserv, \xi, \psi_1,  \tau[i]))$.
It implies that $\exists \tau \allowbreak\in Routes(\lserv, l), i\in \mathbb{N}$, such that $d_{\lserv}^f[l, \tau[i]]\in[d_1,d_2] $ and $\forall x' \in \minre\limits_{j\leq i}(\spresv(\lserv, \xi,\psi_1,\tau[j])),  x'\succre \mathbf{0}$. 
Thus, we have $x''\succre \mathbf{0}$ for all $x''\in \spresv(\lserv, \xi, \psi_1, \tau[j])$ and all $j\leq i$. From the induction hypothesis, $(\lserv, \xi, \tau[j]) \models \psi_1$ for all $j\leq i$. Thus, we have $(\lserv, \xi, l)\models \escape_{[d_1, d_2]}^f \psi_1$.


\vspace{3ex}
\noindent (3)~$(\lserv, \xi, l) \models \psi \Longrightarrow \exists x \in \spresv(\lserv, \psi, \xi, l) \; s.t. \; x\succre \mathbf{0} \; or \;  x\eqre \mathbf{0}$
        
\noindent\textbf{Case} $\psi = \sr_{d_1,d_2}(\varphi)$.
$(\lserv, \xi, l) \models \neg \varphi \reach_{[0,d_1]} (\varphi \reach_{[d_2, \infty)} \varphi)$ implies that there exists a route $\tau\in Routes(\lserv, l)$ and a recovery point $i$ such that $d_\tau^f[i]\in [0,d_1]$ and $(\lserv, \xi, \tau[j])\models\neg\psi$ hold for all location indices $j<i$ and $(\lserv, \xi, \tau[i])\models \psi\reach_{[d_2, \infty)}\psi$. We also know $(\lserv, \xi, \tau[i]) \models \psi\reach_{[d_2, \infty)}\psi$ implies there exists a route $\tau'$ starting from $\tau[i]$ and a location index $i'$ such that $d_{\tau'}^f[i']\in[d_2,\infty)$ and $(\lserv, \xi, \tau[j']) \models \varphi$ hold for all location indices $j'<i'$. To conclude, there exists a route whose prefix is $\tau$ until the location index $i$, and suffix is $\tau'$ until the location index $i'$. The length of the prefix belongs to $[0,d_1]$ and the length of the suffix belongs to $[d_2, \infty)$. Thus, $x_r = d_1-d_{\tau'}^f[i] \geq 0$, $x_p=d^f(\tau'')-d_2\geq 0$. Thus $x\succre \mathbf{0}$ or $x\eqre \mathbf{0}$.



\noindent\textbf{Case} $\psi=\neg \psi_1$.
From the induction hypothesis, we have $(\lserv, \xi, l) \models \neg \psi_1$ implies $\exists x\in \spresv(\lserv, \xi, \psi_1, l)$ s.t. $x\precre \mathbf{0}$ or $x\eqre \mathbf{0}$. By definition ~\ref{sprs quantitative semantics}, we have $(-x_r, -x_p) \in \spresv(\lserv, \xi, \neg\psi, l)$. Thus $(-x_r, -x_p) \succre \mathbf{0}$ or $(-x_r, -x_p) \eqre \mathbf{0}$.

\noindent\textbf{Case} $\psi = \psi_1 \wedge \psi_2$.
$(\lserv, \xi, l) \models  \psi_1 \wedge \psi_2$ implies $(\lserv, \xi, l) \models  \psi_1$ and $(\lserv, \xi, l) \models  \psi_2$. 
From the induction hypothesis, $(\lserv, \xi, l) \models  \psi_1$ implies $\exists x' \in \spresv(\lserv, \xi, \psi_1,l)$ s.t. $x'\succre\mathbf{0}$ or $x'\eqre\mathbf{0}$. Thus $\forall x' \in \spresv(\lserv, \xi, \psi_1,l)$ s.t. $x'\succre\mathbf{0}$ or $x'\eqre\mathbf{0}$.
Similarly, $(\lserv, \xi, l) \models  \psi_2$ implies $\forall x'' \in \spresv(\lserv, \xi, \psi_2,l)$ s.t. $x''\succre\mathbf{0}$ or $x''\eqre\mathbf{0}$.
Thus $\forall x \in \minre (\spresv( \lserv, \xi, \psi_1, l) \cup  \spresv(\lserv, \xi, \psi_2, l)), x\succre \mathbf{0}$ or $x\eqre \mathbf{0}$.

\noindent\textbf{Case} $\psi = \psi_1 \vee \psi_2$.
$(\lserv, \xi, l) \models  \psi_1 \vee \psi_2$ implies $(\lserv, \xi, l) \models  \psi_1$ or $(\lserv, \xi, l) \models  \psi_2$. 
From the induction hypothesis, $(\lserv, \xi, l) \models  \psi_1$ implies $\exists x' \in \spresv(\lserv, \xi, \psi_1,l)$ s.t. $x'\succre\mathbf{0}$ or $x'\eqre\mathbf{0}$. Thus $\forall x' \in \spresv(\lserv, \xi, \psi_1,l)$ s.t. $x'\succre\mathbf{0}$ or $x'\eqre\mathbf{0}$. Similarly, $(\lserv, \xi, l) \models  \psi_2$ implies $\forall x'' \in \spresv(\lserv, \xi, \psi_2,l)$ s.t. $x''\succre\mathbf{0}$ or $x''\eqre\mathbf{0}$. Thus $\forall x \in \maxre (\spresv( \lserv, \xi, \psi_1, l) \cup  \spresv(\lserv, \xi, \psi_2, l)), x\succre \mathbf{0}$ or $x\eqre \mathbf{0}$.

\noindent\textbf{Case} $\psi = \psi_1 \reach^f_{[d_1, d_2]} \psi_2$.
We have $(\lserv, \xi, l) \models \psi_1 \reach^f_{[d_1,d_2]} \psi_2$ implies that $\exists \tau\in Routes(\allowbreak\lserv, $ \ $ l), i \in \mathbb{N}$ s.t. $d_{\tau}^f[i]\in[d_1, d_2]$, $(\lserv, \xi, \tau[i]) \models \psi_2 $ and $ \forall j<i, (\lserv, \xi, \tau[j]) \models \psi_1$. 
From the induction hypothesis, $\exists x\in \spresv(\lserv, \xi, \psi_2, \tau[i]),$ s.t. $x\succre\mathbf{0}$ or $x\eqre\mathbf{0}$ and $\forall j<i, \exists x\in\spresv(\lserv, \xi, \psi_1, \tau[j])$ s.t. $x\succre\mathbf{0}$ or $x\eqre\mathbf{0}$. Thus, $\forall x\in \spresv(\lserv, \xi, \psi_2, \tau[i])$,  $x\succre\mathbf{0}$ or $x\eqre\mathbf{0}$, and $\forall j<i, \forall x\in\spresv(\lserv, \xi, \psi_1, \tau[j]), x\succre\mathbf{0}$ or $x\eqre\mathbf{0}$. Thus $\exists x\in \minre(
            \spresv(  \lserv, \xi, \psi_2,  \tau[i]), $ \ $
            \minre\limits_{j<i}(
                \spresv(  \lserv, \xi, \psi_1,  \tau[j])) )
                $ s.t. $x\succre\mathbf{0}$ or $x\eqre\mathbf{0}$. 

\noindent\textbf{Case} $\psi = \escape_{[d_1,d_2]}^f \psi_1$.
$(\lserv,\xi,l)\models \escape_{[d_1,d_2]}^f \psi_1$ implies $\exists \tau\in Routes(\lserv, l), \allowbreak i\in \mathbb{N}$ s.t. $d_{\lserv}^f[l,\tau[i]] \in [d_1,d_2] $ and $ \forall j\leq i, (\lserv,\xi,\tau[j])\models\psi_1$. From the induction hypothesis, $\forall j\leq i, \exists x\in\spresv(\lserv,\xi,\psi_1,\tau[j])$ s.t. $x\succre\mathbf{0}$ or $x\eqre\mathbf{0}$. Thus $\exists x\in\minre\limits_{j\leq i}\spresv(\lserv, \xi, \psi_1, \tau[j])$ s.t. $x\succre\mathbf{0}$ or $x\eqre\mathbf{0}$. 
\end{proof}

\section{Evaluation Algorithm}\label{sec:evaluation}

\begin{algorithm}[t]%
  \caption{Evaluate($\lserv$, $\xi$, $\psi$)}
  \label{algo:top_evaluate}
\begin{algorithmic}[1]

\CASE{$\psi=S_{d_1, d_2} (\varphi)$}
\STATE $\pcsts=\text{EvalS$(\lserv,f, d_1,d_2, \varphi)$}$
\RETURN{ $\pcsts$ }
\ENDCASE

\CASE{$\psi=\neg\psi_1$}
\STATE $\pcsts_1=\text{Evaluate}(\lserv,\xi,\psi_1)$
\STATE $\pcsts=[]$
\FORALL{ $\ell \in L$ } 
\STATE $\pcsts(\ell)=\{(-x_r,-x_p) \,|\, (x_r,x_p)\in\pcsts_1(\ell)\}$
\ENDFOR 
\RETURN{ $\pcsts$ }
\ENDCASE

\CASE{$\psi=\psi_1 \wedge \psi_2$}
\STATE $\pcsts_1=\text{Evaluate}{(\lserv,\xi,\psi_1)}$
\STATE $\pcsts_2=\text{Evaluate}{(\lserv,\xi,\psi_2)}$
\STATE $\pcsts=[]$
\FORALL{ $\ell \in L$} 
\STATE $\pcsts(\ell)=\minre(\pcsts_1(\ell)\cup\pcsts_2(\ell))$
\ENDFOR
\RETURN{ $\pcsts$ }
\ENDCASE

\CASE{$\psi=\psi_1 \reach_{[d_1,d_2]}^{f} \psi_2$}
\STATE $\pcsts_1=\text{Evaluate}(\lserv,\xi,\psi_1)$
\STATE $\pcsts_2=\text{Evaluate}(\lserv,\xi,\psi_2)$
\STATE $\pcsts=\text{EvalReach}(\lserv,f,d_1,d_2,\pcsts_1,\pcsts_2)$
\RETURN{ $\pcsts$ }
\ENDCASE

\CASE{$\psi=\escape_{[d_1, d_2]}^{f} \psi_1$}
\STATE $\pcsts_1=\text{Evaluate}(\lserv,\xi,\psi_1)$
\STATE $\pcsts=\text{EvalEscape}(\lserv, f, d_1, d_2,\pcsts_1)$
\RETURN{ $\pcsts$ }
\ENDCASE

\CASE{$\psi=\diamonddiamond_{[d_1, d_2]}^{f} \psi_1$}
\STATE $\pcsts_1=\text{Evaluate}(\lserv,\xi,\psi_1)$
\STATE $\pcsts=\text{EvalSomewhere}(\lserv, f, d_1, d_2,\pcsts_1)$
\RETURN{ $\pcsts$ }
\ENDCASE

\CASE{$\psi=\boxbox_{[d_1, d_2]}^{f} \psi_1$}
\STATE $\pcsts_1=\text{Evaluate}(\lserv,\xi,\psi_1)$
\STATE $\pcsts=\text{EvalEverywhere}(\lserv, f, d_1, d_2,\pcsts_1)$
\RETURN{ $\pcsts$ }
\ENDCASE

\end{algorithmic}
\end{algorithm}

We present algorithms for evaluating SpaRS formulas over an A-spatial model $\lserv=(L,W)$. 
The algorithm computes the SpaRV of a SpaRS formula, i.e., a set of $(x_r, x_p)$ pairs, for each starting location \(s\in L\). 
Algorithm~\ref{algo:top_evaluate} is designed by induction on the syntax of the SpaRS formula.

\begin{algorithm}[t]
  \caption{EvalS$((L,W), \, f:A\!\to\!\mathbb{R}_{\geq0}, \, d_1\geq0, \, d_2\geq0, \, \varphi:\text{SREL formula})$ }
  \label{alg:EvalS}
  \begin{algorithmic}[1]
    \STATE Compute $L_{\Phi},L_{\neg\Phi}$ with SREL Monitor.
    \STATE Initialize $F[s]=\emptyset$ for all $s\in L$.
    \FOR{each $v\in L_{\Phi}$}
        
        \STATE \textbf{// Stage 1: recoverability from all $s$ to $v$}
        \STATE $U = L_{\neg\Phi}\cup\{v\}$ 
        \STATE $W[U] = \{(u,w,nbr)\in W \mid u\in U,\ nbr\in U\}$.
        \STATE Initialize $\mathrm{rec}[u]= \infty$ for all $u\in L_{\neg\Phi}$ ; $\mathrm{rec}[v]= 0$ 
        \STATE Priority queue $Q=\{(0,v)\}$.
        \WHILE{$Q\not=\emptyset$ }
          \STATE Pop $(d,u)$ with smallest $d$; \textbf{continue} if $d\neq \mathrm{rec}[u]$.
          \FOR{each $(u,w,nbr)\in W[U]$ with $u$ fixed} 
            \IF{$d+w < \mathrm{rec}[nbr]$}
              \STATE $\mathrm{rec}[nbr]= d+w$; \quad push $(\mathrm{rec}[nbr],\ nbr)$ into $Q$
            \ENDIF
          \ENDFOR
        \ENDWHILE

      \STATE \textbf{// Stage 2: persistency starting from $v$ in $\Sigma_\varphi$}
      \STATE $d_v = \displaystyle\sum_{(u,w,x)\in W[L_v]} f(w)$
        \STATE \textbf{if} $\big(\forall u\in L_v:\ \deg(u)= 0 \pmod 2\big)$ \textbf{ then return } $d_v$ \quad 
        \STATE \textbf{else if} $\big(|\{u\in L_v\ | \ \deg(u)=1 \pmod 2\}|=2 \ \wedge\ \deg(v)=1\big)$ \textbf{ then return } $d_v$ \quad 
        \STATE $per= 0$ 
        \STATE LIFO stack $S=(v,\ \mathbf{0}_m,\  0)$ 
        \WHILE{$S\not=\emptyset$}
          \STATE $(u,used,acc)$ = $S$.pop(); 
          \STATE $per = \max(per,acc)$.
          \FOR{each $(x,\,k,\,f)\in \mathsf{adj}_v(u)$ with $used_k=0$}
            \STATE push $\big(x,\ used \lor \mathbf{e}_k,\ acc + f\big)$ \quad 
          \ENDFOR
        \ENDWHILE
    \STATE \textbf{// Stage 3: combiner}    
    \STATE $F[v] = (d_1, per-d_2)$
        
    \FOR{each $s\in L_{\neg\Phi}$}
        \STATE $F[s]=\maxre(F[s], (d_1-\mathrm{rec}[s],per-d_2))$
      \ENDFOR
    \ENDFOR
    \STATE \textbf{return} $F$
  \end{algorithmic}
\end{algorithm}


\subsection{Evaluation of $S_{d_1,d_2}(\varphi)$}\label{sec:evaluate-S}

We design a two-stage Algorithm~\ref{alg:EvalS} for computing the SpaRV of
$S_{d_1,d_2}(\varphi)$. The SpaRV consists of all maximal resilience pairs $(x_r, x_p)$, where each pair corresponds to a recovery route followed by a persistency route. To ensure maximality, the recovery component should be as short as possible, while the persistency component should be as long as possible.

Our procedure considers every node $v$ satisfying $\varphi$ as a candidate recovery point and computes its associated $(x_r,x_p)$ pair. 
Stage~1 computes the shortest recoverability distance from each starting location to $v$ using Dijkstra’s algorithm. Stage~2 computes the longest persistency distance reachable from $v$ by 
performing a depth-first search over the graph.
By evaluating every $\varphi$-satisfying node as a recovery node and collecting the resulting $(x_r, x_p)$ pairs, the algorithm derives the  maximal resilience set for $S_{d_1,d_2}(\varphi)$.

\begin{algorithm}[t]
  \caption{\text{EvalReach}$((L,W),   f{:}\, A\!\to\!\mathbb{R}_{\geq0},  d_1\geq0,  d_2\geq0 , \, s_1{:}\,L\!\to\!\{(x_r,x_p)\}, \; s_2{:}\,L\!\to\!\{(x_r,x_p)\})$}
  \label{alg:reach}
  \begin{algorithmic}[1]

    \STATE \textbf{Init:} For all $\ell\in L$,
       \[
         \ssign[\ell] = \begin{cases}
           s_2[\ell], & d_3=0\\
           \emptyset, & \text{otherwise}
         \end{cases}
         \qquad
         Q = \{\,(\ell,\, s_2[\ell],\, 0)\ :\ \ell\in L\,\}.
       \]

    \WHILE{$Q \neq \emptyset$} 
      \STATE $Q' = \emptyset$ 
      \FORALL{$(\ell,\,v,\,d)\in Q$}
        \FORALL{ $\ell':\ell' \stackrel{w}{\longrightarrow} \ell$}
          \STATE $\ v' = \minre\!\big( v \cup s_1[\ell'] \big)$
          \STATE $d' = d + f(w)$
          \IF{$d_1 \le d' \le d_2$}
            \STATE $\ssign[\ell'] = \maxre\!\big( \ssign[\ell'] \cup v' \big)$ 
          \ENDIF
          \IF{$d' < d_2$}

          \IF{$\exists (\ell',v'',d')\in Q'$} 
        \STATE $Q' = (Q'-\{ (\ell',v'',d') \})\cup \{ (\ell',\maxre(v'\cup v''),d') \}$
        \STATE \textbf{else} {$Q' = Q'\cup \{ (\ell',v',d') \}$}
        \ENDIF
\ENDIF

        \ENDFOR
      \ENDFOR
      \STATE $Q = Q'$
    \ENDWHILE
    \STATE \textbf{return } $\ssign$
  \end{algorithmic}
\end{algorithm}

To make the expression easier, we define symbols for specific subgraphs associated with a SpaRS formula $\varphi$. We define
\begin{align*}
  L_{\Phi} \;=\; \{l\in L \mid \fmon( \lserv, \xi, \varphi, l)=\top\},\qquad
  L_{\neg\Phi} \;\triangleq\; L\setminus L_{\Phi}.
\end{align*}

For $S\subseteq L$, the induced edge set is
\[
  W[S] \;=\; \{\, (u,w,v)\in W \mid u\in S \ \land\ v\in S \,\}.
\]

We denote $\lserv_\Phi=(L_\Phi, W[L_\Phi])$.



\noindent \textbf{Stage 1.} Given a node $v\in L_\Phi$, the goal is to find the shortest route starting from node $s$ and ending at node $v$, along which all nodes violate $\varphi$ except node $v$, for all $s\in L_{\neg\Phi} $.
This is computed by a single Dijkstra run on the induced subgraph $W[L_{\neg\Phi}\cup\{v\}]$ with source $v$.

\noindent \textbf{Stage 2.} Let $L_v$ and $W[L_v]$ be the sets of nodes and edges of the connected component of $\lserv_{\Phi}$ that contains $v$, and let $\deg(u)$ be the degree of the nodes $u\in L_v$ in the A-spatial model $\lserv_v=(L_v, W[L_v])$.
The goal is to find the distance of the longest route starting from $v$ in $\lserv_v$.
Let 
\(
d_v= \sum_{(u,w,x)\in W[L_v]} f(w).
\)
Assign a bijection (edge indexing)
\(
\kappa: W[L_v]\to \{0,1,\dots,m-1\}.
\)
We represent edge usage by a bit-vector $used\in\{0,1\}^m$, where $used_k=1$ iff the undirected edge with index $k$ has been taken (in either direction).
For the adjacency of each node $u\in L_v$, define
\(
\mathsf{adj}_v(u)\ =\ \big\{\,\big(x,\ \kappa(\{u,x\}),\ f(w)\big)\ |\ \{u,w,x\}\allowbreak \in W[L_v]\,\big\}.
\)
where $\kappa$ indexes undirected edges; $used_k=1$ forbids reusing edge $k$ in either direction. $\mathbf{e}_k$ is the length-$m$ unit vector with a $1$ in position $k$. 
We first check whether $\lserv_v$ admits an Eulerian circuit or an Eulerian trail starting at $v$. If either condition holds, then the longest route starting from $v$ traverses every edge exactly once, and its length is $d_v$. Otherwise, we perform a depth-first search over $\lserv_v$, respecting the edge-usage mask $used$, to enumerate all edge-simple routes from $v$, and return the maximum accumulated distance encountered.




\noindent\textbf{Combiner.} For each recovery node $v\in L_{\Phi}$, Stage~1 returns recoverability for all starting locations $s\in L_{\neg\Phi} \cup \{v\}$; Stage~2 returns the length of the persistency route starting from $v$.
For any starting location $s\in L_{\neg\Phi}$, we return the resilience maximal pairs over all recovery nodes $v\in L_{\Phi}$. For any starting location $s\in L_\Phi$, since recovery occurs on $s$ itself, we return $(x_r,x_p)=(d_1, per-d_2)$.

\paragraph{Correctness.}
Every satisfying route from $s$ uniquely factors as a violation-only prefix ending at the first true node $v$ and a satisfaction-only suffix starting at $v$.
Stage~1 computes the minimal length of the route ending at $v$ under the ``first satisfaction node is $v$'' constraint; Stage~2 computes the maximal length of the route starting from $v$.
Therefore, the pair $(d_1-\mathrm{rec},\ \mathrm{per}-d_2)$ exactly matches the S-atom margins, and $\maxre$ returns the semantics at $s$.

\begin{algorithm}[tbp] 
\caption{EvalEscape($(L,\wfun)$, $f: A\rightarrow \mathbb{R}_{\geq0}$, $d_1\geq0$,$d_2\geq 0$, $s_1: L\rightarrow {(x_r, x_p)}$ )}
\label{alg:escape}
\vspace{1mm}
\begin{algorithmic}[1]
\STATE $D = \text{MinDistance}{(L,\wfun,f)}$
\STATE $\forall \ell,\ell'\in L. e[\ell,\ell'] = {(-\infty, -\infty)}$
\STATE $\forall \ell\in L. e[\ell,\ell]=s_1(\ell)$
\STATE $T=\{ (\ell,\ell) | \ell\in L \}$
\WHILE{ $T\not= \emptyset$ }
\STATE $e'=e$ 
\STATE $T'=\emptyset$
\FORALL{ $(\ell_1,\ell_2) \in T$ }
\FORALL{ $\ell_1': \nextto{\ell_1'}{w}{\ell_1}$ }
\STATE{ $v = \maxre(e[\ell_1',\ell_2], \minre(s_1(\ell_1'),  e[\ell_1,\ell_2]))$}
\IF{$v\not=e[\ell_1',\ell_2]$}
\STATE{$T'=T'\cup \{ (\ell_1',\ell_2) \}$}
\STATE{$e'[\ell_1',\ell_2]=v$}
\ENDIF
\ENDFOR
\ENDFOR
\STATE{T=T'}
\STATE{e=e'}
\ENDWHILE
\STATE $\ssign=[]$
\FORALL{ $\ell\in L$ }
\STATE $\ssign(\ell)=\maxre(\{ e[\ell,\ell'] | D[\ell,\ell']\in [d_1,d_2] \})$
\ENDFOR
\RETURN $\ssign$
\end{algorithmic}
\end{algorithm} 

\paragraph{Complexity.}
Stage~1 runs Dijkstra once per $v\in L_{\Phi}$ on $W[L_{\neg\Phi}\cup\{v\}]$, thus a complexity of $O(|E|\log|V|)$.
Stage~2 is linear when an Eulerian shortcut applies; 
otherwise, it is exponential in $|E[L_v]|$ (exhaustive DFS).

\subsection{Spatial operators} For the reach and escape operators, we adopt the flooding algorithm of~\cite{2022STREL}. In particular, we use Algorithm~\ref{alg:reach} to evaluate the reach operator and Algorithm~\ref{alg:escape} to evaluate the escape operator. Since our semantics is defined as a set of non-dominated $(x_r, x_p)$ pairs, we change max and min to $\maxre$ and $\minre$, respectively. 

Algorithm~\ref{alg:reach} evaluates the reach operator using the flooding
procedure of~\cite{2022STREL}. For each node $\ell$, the algorithm keeps a
set $s(\ell)$ of non-dominated resilience pairs. A queue $Q$
initializes all nodes with starting semantics $s_2(\ell)$.
During the flood, each entry $(\ell, v, d)$ is propagated
backwards along incoming edges $\ell' \xrightarrow{w} \ell$.
The updated distance $d' = d + f(w)$ is discarded if $d' > d_4$.
Otherwise, the new semantics $v' = \minre(v \cup s_1(\ell'))$ is valid.
If $d_1 \le d' \le d_2$, the result at $\ell'$ is updated to $\maxre(s(\ell') \cup v')$.
If $d' < d_2$, the tuple $(\ell', v', d')$ is reinserted into the queue,
continuing the flooding until convergence.
The final result $s$ returns semantics at each location.

Algorithm~\ref{alg:escape} evaluates the escape operator, again using the flooding
strategy of~\cite{2022STREL}. Unlike reach, the escape distance constraint
is applied to the shortest-path distance
$d_{\Sigma}(\ell, \ell') \in (d_3,d_4]$, not the accumulated flood
distance. Propagation proceeds only through nodes satisfying the
escape guard. For each predecessor $\ell'$, the updated semantics
$v'=\minre(v \cup s_1(\ell'))$ is added to the result via
$\maxre$ whenever the shortest-path constraint holds.
Flooding continues while $d' < d_4$.
The algorithm thus collects all non-dominated resilience values
associated with valid distance constraint.

For the somewhere and everywhere operators, we can also apply the flooding algorithm to evaluate. It is easy to adapt the Algorithm~\ref{alg:reach} for these two operators. Due to space limit, we only describe in text without pseudo code. For the somewhere, the input is $(L,W),f,d_1,d_2,s_1$. In line 1, we initialize $\ssign$ and $Q$ with $s_1$ instead of $s_2$. In line 6, we change $\minre$ to $\maxre$ since somewhere operator considers the maximum resilience pairs. In this way, we can evaluate somewhere with modified Algorithm~\ref{alg:reach}. For the everywhere operators, we didn't change line 6. Instead, in line 9 and line 13, we change $\maxre$ to $\minre$ since everywhere operator considers the minimum resilience pairs.

\begin{figure*}[t]
    \centering
    \begin{subfigure}[t]{0.24\linewidth}
        \centering
        \includegraphics[width=\linewidth]{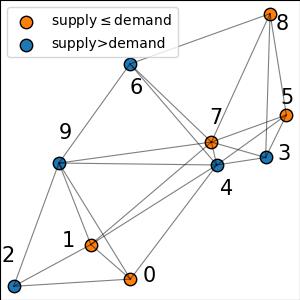}
        \caption{Networked Microgrid}
        \label{fig:Locations and Spatial Signal-NMG}
    \end{subfigure}%
    \hfill
    \begin{subfigure}[t]{0.24\linewidth}
        \centering
        \includegraphics[width=\linewidth]{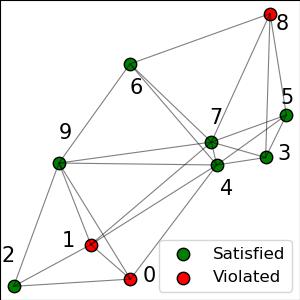}
        \caption{$\psi_1 = \sr_{1000,2000} (\varphi)$}
    \end{subfigure}
    \hfill
    \begin{subfigure}[t]{0.24\linewidth}
        \centering
        \includegraphics[width=\linewidth]{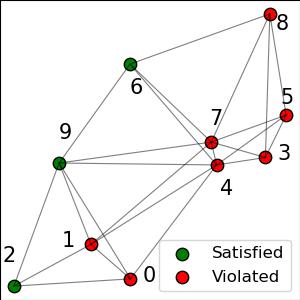}
        \caption{$\psi_2=\sr_{1000,1000} (\boxbox_{[0,1000]} \varphi)$}
    \end{subfigure}
    \hfill
    \begin{subfigure}[t]{0.24\linewidth}
        \centering
        \includegraphics[width=\linewidth]{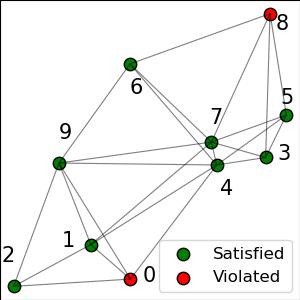}
        \caption{$\psi_3 = \diamonddiamond_{[0,1500]} \sr_{1000,2000}(\varphi)$}
    \end{subfigure}
    \caption{In the NMG model, each node represents an individual microgrid, uniquely identified by a corresponding index.}
    \label{fig:mc-example}
\end{figure*}

\section{Case studies}
We present two case studies: networked microgrids and a bike sharing system, to illustrate the usage of SpaRS specifications. 
We implemented our evaluation algorithm (Algorithms~\ref{algo:top_evaluate}-\ref{alg:escape}) in Python on a computer with an Intel Core i7-14700K CPU, 32 GB DDR5 memory, and Windows 11 operating system.

\subsection{Networked Microgrids}
\begin{figure*}[h]
    \centering
    \begin{subfigure}[t]{0.35\linewidth}
        \centering
        \includegraphics[width=\linewidth]{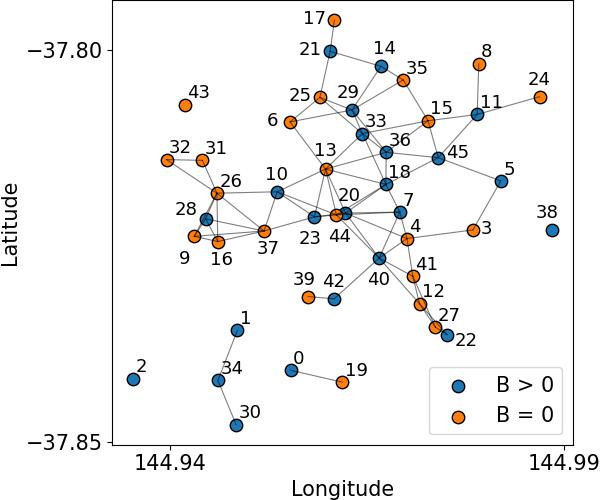}
        \caption{Bike docking stations in the city of Melbourne.}
        \label{fig:bike-map}
    \end{subfigure}
    \hspace{0.1\linewidth}
    \begin{subfigure}[t]{0.35\linewidth}
        \centering
        \includegraphics[width=\linewidth]{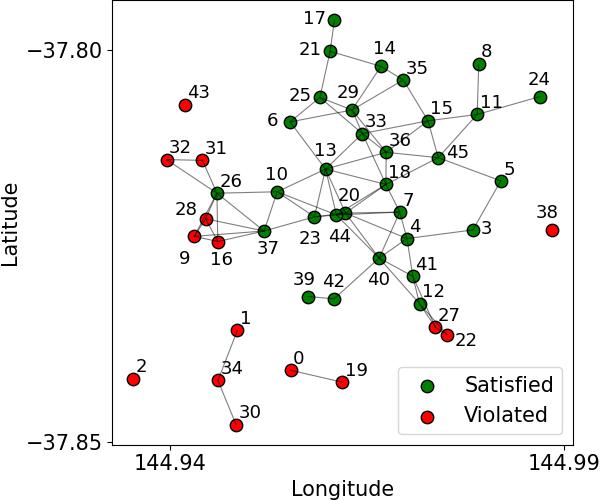}
        \caption{$\psi_1 = \sr_{1000,2000} (\varphi_1)$}
        \label{fig:bike-psi1}
    \end{subfigure}

    \vspace{1em} 

    \begin{subfigure}[t]{0.35\linewidth}
        \centering
        \includegraphics[width=\linewidth]{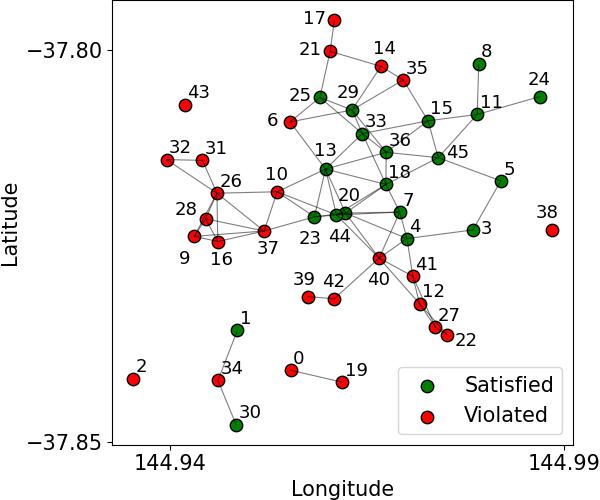}
        \caption{$\psi_2 = \sr_{1000,1000}(\boxbox_{[0,500]} \varphi_1)$}
        \label{fig:bike-psi2}
    \end{subfigure}
    \hspace{0.1\linewidth}
    \begin{subfigure}[t]{0.35\linewidth}
        \centering
        \includegraphics[width=\linewidth]{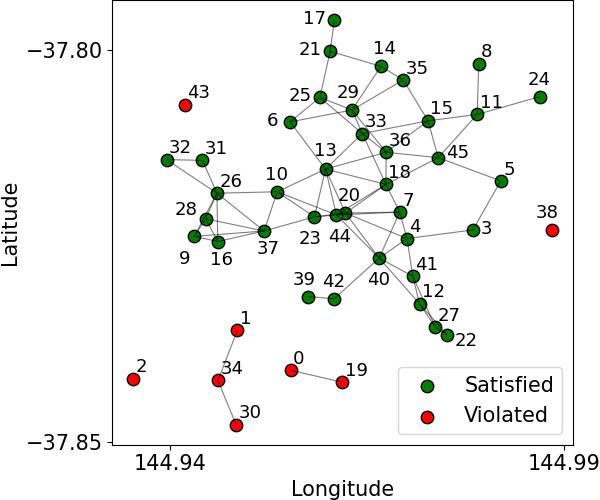}
        \caption{$\psi_3 = \diamonddiamond_{[0,1000]} \sr_{1000,2000}(\varphi_1)$}
        \label{fig:bike-psi3}
    \end{subfigure}

    \caption{In the graphs, each node represents a bike docking station, uniquely identified by a corresponding index.}
\end{figure*}
Microgrids interconnect within diverse power distribution systems to create Networked Microgrids (NMGs), which are designed to exchange electrical power, improve stability, optimize power flow, and provide a sustained supply to areas impacted by outages.
In a mesh topology, multiple neighboring microgrids in a region connect to form a network that allows for power sharing among all members~\cite{Islam2021microgrid, Aghmadi2023microgrid}. 
Microgrids communicate and exchange power via redundant electrical or communication links. 
Each microgrid in an NMG relies on its own distributed energy resources (DERs) to meet local demand. 
If a fault compromises a microgrid's ability to meet its load, it can import power from other interconnected microgrids to facilitate recovery. 
Furthermore, this support mechanism can maintain a power-sufficient state over long distances by leveraging a series of microgrids along the underlying transmission route.

Figure~\ref{fig:mc-example}(a) depicts an a-spatial model of an NMG comprising 10 microgrids. 
The node indices are shown in the figure. 
The edge weights (not displayed in the figure) are (0,1): 821.38, (0,2): 1765.03, (0,4): 2273.58, (0,9): 2181.11, (1,2): 1349.47, (1,4): 2301.07, (1,7): 2471.35, (1,9): 1416.06, (2,9): 2126.95, (3,4): 747.33, (3,5): 761.37, (3,7): 856.27, (3,8): 2351.09, (4,5): 1329.83, (4,6): 2106.47, (4,7): 382.98, (4,9): 2380.34, (5,7): 1214.07, (5,8): 1674.01, (6,7): 1770.92, (6,8): 2266.49, (6,9): 1941.83, (7,8): 2276.86, (7,9): 2324.59.
Let $\varphi=supply\geq demand$ specifies that supply is greater than or equal to demand at a given microgrid.
In Figure~\ref{fig:mc-example}(a), microgrids satisfying and violating this condition are indicated by a specific color code, respectively. 
We proceed to evaluate the following SpaRS formulas.

Formula $\psi_1=\sr_{d_1,d_2}(\varphi)$ specifies that there exists a microgrid that satisfies $supply\geq demand$ within distance $d_1$, and this supply-efficient state can be maintained for at least distance $d_2$.

Formula $\psi_2= \sr_{d_3,d_4} (\boxbox_{[0,d_5]}\varphi)$ specifies that there exists a microgrid that satisfies $\boxbox_{[0,d_5]} supply$ $\geq demand$ within distance $d_3$, and this property can be maintained for at least distance $d_4$.

Formula $\psi_3= \diamonddiamond_{[0,d_6]} \sr_{d_7,d_8}(\varphi)$ specifies that there exists a microgrid within distance $d_6$ from location $l$ which satisfies $\sr_{d_7,d_8}(\varphi)$.

\begin{table}[t]
\setlength{\tabcolsep}{3.5pt}
    \centering
    \begin{tabular}{|c|c|c|c|}
    \hline
    & 1 & 4 & 5 \\ \hline
    $\psi_1$ & \{(-349.47,6555.60)\} & \{(1000,5175.97)\} & \{(238.63,5175.97)\}  \\  \hline
    $\psi_2$ & \{(-349.47,3068.79)\} & \{(-1106.47,3068.79)\} & \{(-1984.99,3068.79)\} \\ \hline
    $\psi_3$ & \{(1000,6555.60)\} & \{(1000,5175.97)\} & \{(1000,5175.97)\}   \\ \hline
    \end{tabular}
    \caption{SpaRV values of $\psi_1, \psi_2, \psi_3$ at locations $1,4,5$}
    \label{tab:mc-values}
    \vspace{-0.8cm}
\end{table}


We show the truth values of $\psi_1, \psi_2, \psi_3$ at each microgrid in Figure~\ref{fig:mc-example}(b)-(d). 
For a location, if the calculated semantics include a pair $(x_r, x_p)$, where $x_r\geq0$ and $x_p\geq0$, then $\psi_1$ is satisfied at this location; otherwise, $\psi_1$ is violated.
Table~\ref{tab:mc-values} shows the SpaRVs values at select microgrid locations. 
For the semantics of $\psi_1$ at location~1, the corresponding route is $1\rightarrow2\rightarrow9\rightarrow6\rightarrow4\rightarrow9$. The distance between locations along this route is 1349.47, 2126.95, 1941.83, 2106.47, 2380.34. Thus, $x_r=1349.47-1000=349.47, x_p=2126.95+1941.83+2106.47+2380.34-2000=6555.59$. 
For $\psi_2$, we first use an SREL monitor to evaluate $\boxbox_{[0,1000]} \varphi$. Locations 2, 6, and 9 satisfy $\boxbox_{[0,1000]} \varphi$, whereas the other locations do not. 
We then apply the evaluation algorithm on $\psi_2$. 
For example, location~7 violates $\psi_2$ because it cannot find a route whose distance is less than 1,000 to reach locations 2, 6 or 9; even the shortest route $7\rightarrow6$ is 1,770.92 long. 
For $\psi_3$, location 0 violates it because only location 1 is within distance 1500, but neither location 0 or 1 satisfy $S_{1000, 2000}(\varphi)$. 
The computation times of evaluating SpaRVs for formulas $\psi_1, \psi_2$, and $\psi_3$ are (in seconds) 0.52, 0.84, and 1.47, respectively.  

\subsection{Bike Sharing System}

In an urban bike sharing system, bikes are distributed across a spatial network of docking stations. A station may run out of bikes when the demand exceeds the supply. 
A goal of the bike sharing system is to ensure resilience of service:  if a station is out of bikes, a nearby station with availability should exist (recoverability), and the presence of stations with bike availability should be maintained (post-recovery) for a reasonable distance (persistency).


Let $\varphi_1=B>0$, where $B$ denotes the number of bikes at a station. 
We specify the following SpaRS formulas.

Within distance $d_1$ there exists a station that has bikes, and from that station there exists a route of length at least $d_2$ along which all stations have bikes.   
\begin{equation*}
    \psi_1 = \sr_{d_1,d_2} (\varphi_1) 
\end{equation*}

Within distance $d_3$, one can reach a \emph{well-served station}---i.e., a station such that every station within distance $d_5$ of it has bikes---and
from there, the user can consecutively visit well-served stations along a route of length at least $d_4$.
\begin{equation*}
    \psi_2 = \sr_{d_3, d_4} (\boxbox_{[0,d_5]} \varphi_1 )
\end{equation*}

Within distance $d_6$, there is a station satisfying 
$\psi_1$.
\begin{equation*}
    \psi_3 = \diamonddiamond_{[0,d_6]} \sr_{d_1,d_2}(\varphi_1 )
\end{equation*}

We obtained docking station location data (longitude and latitude) and bike numbers of the city of Melbourne~\cite{melbourne_bike_docks}, and rendered it in  Figure~\ref{fig:bike-map}. We randomly set half of the docking stations to have no bikes, i.e., $B=0$.  Two stations are connected if they are at a distance of at most 800 meters apart. 

Figure~\ref{fig:bike-psi1} shows the truth value of $\psi_1$, where $d_1=1000, d_2=2000$.  
We find that most of the locations satisfy $\psi_1$. For example, for location~24, the best route is $24\rightarrow 11\rightarrow 45\rightarrow 36\rightarrow 29\rightarrow 33\rightarrow 18\rightarrow 7\rightarrow 40\rightarrow 20\rightarrow 23\rightarrow 10$. Location~11 is the first to satisfy $\varphi_1$ along this route. Thus, recoverability is associated with the prefix $24\rightarrow 11$, whose length is $rec=743.85$, whereas persistency is associated with the suffix $11\rightarrow 45\rightarrow \cdots \rightarrow 10$, whose length is $per=5919.27$. Finally, we have $(x_r, x_p)= (1000-rec, per-2000) = (256.15, 3919.27)$. Similarly, Figure~\ref{fig:bike-psi2} and~\ref{fig:bike-psi3} show the truth values of $\psi_2$ and $\psi_3$, respectively.
The computation times of the SpaRVs for formulas $\psi_1, \psi_2$, and $\psi_3$ are (in seconds) 1.12, 1.27, and 2.97, respectively.



\section{Related work}

The concept of spatial resilience permeates many research areas, including ecology~\cite{Allen2016Ecology, Cumming2011Ecology, Chambers2019Ecology}, chemistry~\cite{https://doi.org/10.1111/ele.12897, Zine2021chemistry}, and socio‐economic systems~\cite{doi:10.1177/0309132513518834, Christmann2012}. 
While these studies are domain-specific, this paper adopts a formal method-based approach to model and reason about spatial resilience in CPSs. 
Below, we review related work in two key areas: spatial logic for CPSs and the formal modeling of resilience.



Spatio-temporal logics have been developed to specify and analyze spatial behavior in CPS.  Early efforts, such as Signal Spatio-Temporal Logic (SSTL)~\cite{SSTL}, introduced basic spatial operators (e.g., somewhere, surround) on top of STL~\cite{Donz2010STL, Maler2004STL}. Beyond SSTL, more expressive formalisms have emerged. For example, SpaTeL~\cite{SpaTel} unifies STL with a tree-based spatial logic (TSSL) to describe high-level spatial patterns that evolve over time. In ~\cite{SpaTel}, SpaTeL is designed to compute the probability that a networked system satisfies a spatio-temporal property. Orthogonally, SaSTL~\cite{SaSTL} extends STL with spatial aggregation and counting operators, allowing specifications over distributed agents and regions. 
More recently, Spatio-Temporal Reach and Escape Logic (STREL)~\cite{2022STREL}  introduces  novel spatial operators in the form of reach and escape (see Section~\ref{SEC:SREL}). 
Balakrishnan et al.~\cite{Balakrishnan2025AFA} proposes an alternating finite automata construction for STREL specifications, enabling efficient offline and online monitoring of distributed CPS.
Based on the spatial fragment of STREL, our work introduces the concept of spatial resilience in terms of CPS recoverability and persistency. 

Researchers have formalized resilience in the temporal domain~\cite{Saoud2023Temporal,chen2022resilience,Hassan2024Resilience, Clark2019resilience}.
Saoud et al.~\cite{Saoud2023Temporal} present a notion of CPS resilience based on Linear Temporal Logic (LTL) that quantifies how much disturbance a system can tolerate while still satisfying an LTL specification.  
The STL-based Resilience Specification (SRS) framework proposed by Chen et al.~\cite{chen2022resilience} considers recoverability and durability for temporal signals. 
Hassan et al.~\cite{Hassan2024Resilience} develop Finite-Time Robust Control Barrier Functions (FR-CBFs) that impose explicit guarantees on both recovery time and post-recovery safety invariance for power inverter networks under worst-case disturbances and cyber-attacks.
Clark et al.~\cite{Clark2019resilience} present resilience metric to quantify the ability of the system to recover from an attack provided the attack is discovered within a fixed time interval, as well as the cost of recovery.
In contrast, our proposed SpaRS resiliency framework supports reasoning about recoverability and persistency in a spatial logic framework. 
Chen et al.~\cite{chen2025cumulative} introduced Cumulative-Time Signal Temporal Logic (CT-STL), a concept analogous to our notion of spatial persistency; however, it applies to temporal domains rather than spatial ones, and it measures accumulation over time instead of continuity over distance.
Other research has investigated resilient behavior, both spatial and temporal, for planning and control of cyber-physical systems~\cite{Chen2023ResilienC, Ivanov2016attackresilient, Wu2024resilience, chen2016resilientmicrogrids, chen2020automatica}, including autonomous vehicles, microgrids, and multi-agent systems. The focus of this paper is on the formal specification and reasoning of spatial resilience, while spatial resilient control is a direction for our future work.


\section{Conclusion}
\label{sec:conclusion}

In this paper, we presented a logical framework for reasoning about spatial resiliency in CPS.  We defined the resiliency of an SREL formula $\varphi$ as the ability of the system to recover from violations of $\varphi$ in a spatially efficient and durable manner. These requirements represent the atoms of our
SpaRS logic, which allows one to combine such resiliency primitives using spatial and Boolean operators. We also introduced SpaRV, the first multi-dimensional semantics for an SREL-based logic. Under this semantics, a SpaRS formula is interpreted as a set of non-dominated $(rec, per)$ pairs, which respectively quantify (in spatial units)
how quickly the underlying system recovers from a property violation and for how long it satisfies the property thereafter. Importantly, we proved that our SpaRV semantics is sound and complete w.r.t.\ SREL's Boolean semantics. 
We illustrated our new resiliency framework with two case studies: networked microgrids and an urban bike sharing system. 
Collectively, our results demonstrate the expressive power and flexibility of our framework in reasoning about spatial resiliency in CPS. 

In summary, the contributions of our work lie not just in establishing theoretical foundations of spatial resiliency, but also in providing a method to equip spatial logics with multi-dimensional semantics.
Such an approach could, in the future, be extended to support arbitrary multi-requirement specifications beyond resiliency.

\balance
\bibliographystyle{ACM-Reference-Format}
\bibliography{bib}

\end{document}